\documentclass[8pt,twocolumn]{IEEEtran}

\makeatletter
\def\ifundefined{\@ifundefined}
\makeatother

\usepackage{epsfig}
\usepackage{dsfont}
\usepackage{amsmath}
\usepackage{amssymb} 
\usepackage{graphicx}
\usepackage{amsthm}
\usepackage[normalem]{ulem}
\usepackage{subfigure}
\usepackage{color}
\usepackage[hidelinks]{hyperref}
\usepackage{bm}
\usepackage{lineno}
\usepackage{multirow}
\usepackage[ruled,vlined]{algorithm2e}
\usepackage{amsmath}
\usepackage{cite}
\usepackage{lettrine}
\usepackage{diagbox} %
\usepackage{rotating}
\usepackage[super]{nth}
\usepackage{algorithmicx,algpseudocode}

\DeclareRobustCommand{\varlambda}{\text{\usefont{OML}{txmi}{m}{it}\symbol{"15}}}
\DeclareSymbolFont{myletters}{OML}{ztmcm}{m}{it}
\DeclareMathSymbol{\uplambda}{\mathord}{myletters}{"15}
\DeclareSymbolFont{matha}{OML}{txmi}{m}{it}%
\DeclareMathSymbol{\varv}{\mathord}{matha}{118}
\makeatletter
\let\start@align@nopar\start@align
\let\start@gather@nopar\start@gather
\let\start@multline@nopar\start@multline
\long\def\start@align{\par\start@align@nopar}
\long\def\start@gather{\par\start@gather@nopar}
\long\def\start@multline{\par\start@multline@nopar}
\makeatother
\def\b{\ensuremath\boldsymbol}
\newcommand{\be}{\begin{equation}}
\newcommand{\ee}{\end{equation}}
\newcommand{\bea}{\begin{eqnarray}}
\newcommand{\eea}{\end{eqnarray}}
\newcommand\ba[1]{\left[ \begin{array}{#1}}
	\def\ea{\end{array}\right]}

\newcommand{\bfi}{\begin{figure}}
	\newcommand{\efi}{\end{figure}}

\newcommand{\R}{\mathbb{R}}

\newcommand{\X}{\mathcal{X}}

\newcommand{\x}{\mathbf{x}}

\newcommand{\cmmnt}[1]{}

\makeatletter
\newcommand*{\transpose}{%
	{\mathpalette\@transpose{}}%
}
\newcommand*{\@transpose}[2]{%
	\raisebox{\depth}{$\m@th#1\intercal$}%
}
\makeatother

\providecommand{\norm}[1]{\lVert#1\rVert}

\DeclareMathOperator*{\argmin}{arg\,min}

\newtheorem{thm}{Theorem}  %

\newtheorem{prop}{Proposition}

\newtheorem{defn}{Definition}       %

\expandafter\def\expandafter\normalsize\expandafter{%
	\normalsize%
	\setlength\abovedisplayskip{0pt}%
	\setlength\belowdisplayskip{8pt}%
	\setlength\abovedisplayshortskip{-8pt}%
	\setlength\belowdisplayshortskip{2pt}%
}

\begin{document}
	\title{Spectral Eigenfunction Decomposition for Kernel Adaptive Filtering 
		\thanks{This work was supported by the ONR grant N00014-23-1-2084.}
		\thanks{The authors are with the Computational NeuroEngineering Laboratory, University of Florida, Gainesville, FL 32611 USA  (e-mail: likan@ufl.edu; principe@cnel.ufl.edu).}}
	\author{Kan Li, \IEEEmembership{Member,~IEEE} and Jos\'{e} C. Pr\'{i}ncipe, \IEEEmembership{Life Fellow,~IEEE}}

	\maketitle
	
\begin{abstract}
	Kernel adaptive filtering (KAF) integrates traditional linear algorithms with kernel methods to generate nonlinear solutions in the input space. The standard approach relies on the representer theorem and the kernel trick to perform pairwise evaluations of a kernel function in place of the inner product, which leads to scalability issues for large datasets due to its linear and superlinear growth with respect to the size of the training data. Explicit features have been proposed to tackle this problem, exploiting the properties of the Gaussian-type kernel functions. These approximation methods address the implicitness and infinite dimensional representation of conventional kernel methods. However, achieving an accurate finite approximation for the kernel evaluation requires a sufficiently large vector representation for the dot products. An increase in the input-space dimension leads to a combinatorial explosion in the dimensionality of the explicit space, i.e., it trades one dimensionality problem (implicit, infinite dimensional RKHS) for another (\textit{curse of dimensionality}). This paper introduces a construction that simultaneously solves these two problems in a principled way, by providing an explicit Euclidean representation of the RKHS while reducing its dimensionality. We present SPEctral Eigenfunction Decomposition (SPEED) along with an efficient incremental approach for fast calculation of the dominant kernel eigenbasis, which enables us to track the kernel eigenspace dynamically for adaptive filtering. Simulation results on chaotic time series prediction demonstrate this novel construction outperforms existing explicit kernel features with greater efficiency. 
\end{abstract}
\section{Introduction}
Kernel methods offer an adaptable and robust framework for solving nonlinear problems in signal processing and machine learning. Traditionally, these methods rely on the representer theorem and the kernel trick to perform pairwise evaluations of a kernel function, which leads to scalability issues for large datasets due to its linear and superlinear growth in computational and storage costs with respect to the data size. 

In the standard kernel approach, data points in the input space are mapped using an implicit nonlinear function into a potentially infinite-dimensional inner product space known as a reproducing kernel Hilbert space (RKHS). The explicit representation is ignored, as inner products are carried out using a real-valued similarity function, called a reproducing kernel. This offers an elegant solution for classification, clustering, and regression, as mapped data points become linearly separable in the RKHS, allowing classic linear methods to be directly applied. However, since points (functions) in the this high-dimensional space are not explicitly accessible, kernel methods scale poorly to large datasets. 

In the case of online kernel adaptive filtering (KAF) algorithms \cite{Liu10}, this implicit mapping behaves as a rolling sum with linear, quadratic, or cubic growth over time. Consequently, much research has focused on reducing the computational load through sparsification techniques \cite{QKLMS, NICE, SNIPGOAL}.

An alternative approach to combat this is to define an explicit mapping or feature where the kernel evaluation is approximated using the dot product in a higher finite-dimensional space or Euclidean space. This framework leverages the continuous shift-invariant properly-scaled kernel function, particularly the Gaussian kernel. One popular method uses random Fourier features (RFF) \cite{rahimi2007RFF} to approximate the kernel, enabling scalable linear filtering techniques to be applied directly on the explicitly transformed data without the computational drawback of the naive kernel method. Although this approach has been applied to various kernel adaptive filtering algorithms \cite{Singh12, Qin17, Bouboulis18}, the approximation quality is still an active research area, as it relies on random sampling with various error bounds proposed \cite{rahimi2007RFF,Sutherland2015,Sun2018,pmlr-v97-li19k}. Notably, RFF methods exhibit higher variance and suboptimal error bounds for the Gaussian kernels compared to other methods \cite{Sutherland2015, Li2019notrick}. 

Alternatively, deterministic feature mappings such as Gaussian quadrature (GQ) \cite{Dao2017} and Taylor series (TS) expansion \cite{Zwicknagl2009, cotter2011explicit}, which is related to the fast Gauss transform in kernel density estimation \cite{Greengard1991, Yang2003}, offer exact-polynomial approximations. This eliminates the undesirable effects of performance variance using random features. We have demonstrated superior performances in KAF using deterministic polynomial features compared to random features \cite{Li2019notrick}.

These approximation methods address the implicitness and infinite dimensional representation of conventional kernel methods. However, achieving an accurate finite approximation for the kernel evaluation requires a sufficiently large vector representation for the dot products. An increase in the input-space dimension leads to a combinatorial explosion in the dimensionality of the explicit space, i.e., it trades one dimensionality problem (implicit, infinite dimensional RKHS) for another (\textit{curse of dimensionality}). In this paper, we present a principled approach that tackles both dimensionality problems simultaneously by providing and maintaining a compact Euclidean representation of the RKHS while systematically reducing its dimensionality.

Principal component analysis (PCA) \cite{Pearson1901, loeve1960probability} plays a crucial role in statistical signal processing such as communications, image processing, and machine learning. PCA enables signal representation in a much lower dimensional subspace, even though the observations lie in a high dimensional space. Here, we apply the same principle to the RKHS and define a lower dimensional subspace using a SPEctral eigenfunction decomposition (SPEED). We find the feature coordinates of each data sample using an eigenmap and directly apply classical linear adaptive filtering methods such as least mean squares (LMS) and recursive least squares (RLS) in this explicit space spanned by the dominant eigenfunctions.

This approach is related to the kernel principal component analysis (KPCA) \cite{KPCA1998} and the kernel principal component regression (KPCR) \cite{KPCR}. The spectrum of the kernel (Gram) matrix is connected to the spectrum of the integral operator for the RKHS associated with the reproducing kernel by a normalization factor \cite{williams2000effect, bengio2003spectral, rosasco10a}. Unlike the previous work, in this paper, we following an adaptive filtering paradigm and present efficient online algorithms for the fast calculation and dynamic tracking of the dominant kernel eigenbasis, along with the corresponding weights of linear models constructed in this finite-dimensional space spanned by the eigenfunctions. 

We do not compute the eigenfunctions explicitly, but rather the projections of the data onto those components, which provides an explicit coordinate system that we can construct linear models on. Finding the eigenfunctions is equivalent to finding the coefficients in the RKHS spanned by the sample basis functions, and the projection onto the dominant eigenfunctions is the dimension-reduction eigenmap we use to derive a lower-dimensional Euclidean representation of the infinite-dimensional RKHS. This approach not only saves computation and memory storage, but also avoids multicollinearity and helps reduce overfitting. This concept is illustrated in Fig. \ref{fig:SPEED}.

\begin{figure}[t!]
	\centering
	\includegraphics[width=0.48\textwidth]{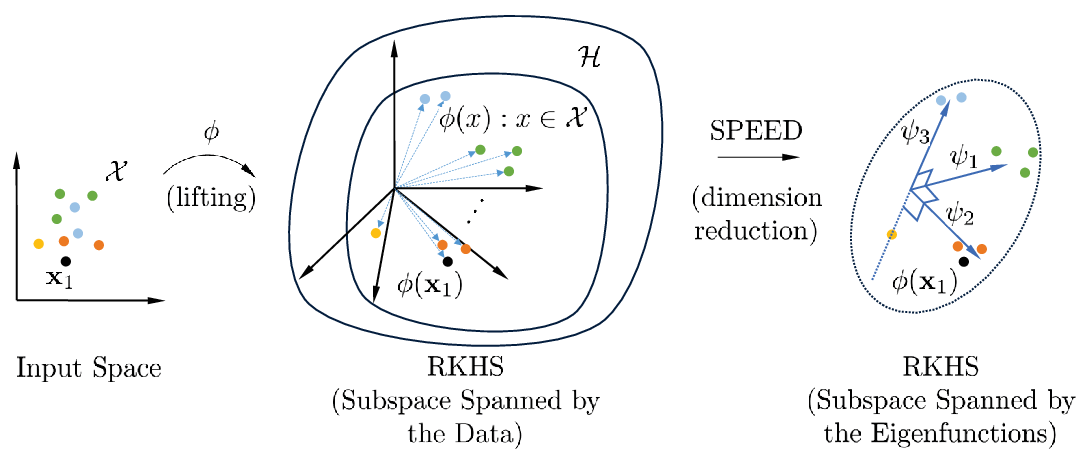}
	\caption{Eigendecomposition in the RKHS.}
	\label{fig:SPEED}
\end{figure}

Unlike previous formulations of explicit feature space mappings, the spectral embedding approach does not rely on a Gaussian-type kernel. Rather, it is data-dependent and kernel-independent, providing a more efficient representation compared to those based on GQ, TS polynomials, or RFF. In contrast, Gaussian approximations are data-agnostic, aiming to ensure convergence for all points (covering the support of the exponential function), which makes them a less efficient choice for representing specific datasets, leading to the curse of dimensionality. SPEED is a principled approach to embed the data statistics directly into the kernel definition for a lower-dimensional representation. 

However, as a batch method, it is not suited for tracking non-stationary input data, such as in adaptive filtering. The Gram matrix size increases quadratically with the number of sample points, as more signal is observed and collected. For a similar algorithm such as KPCA, this results in a significantly high computational cost. One solution is to apply the kernel Hebbian (KHA) algorithm \cite{KHA2005}. While iterative algorithms such as KHA are shown to be computationally efficient, they are not fully \textit{online} in the sense of adaptive signal processing. Other exact incremental KPCA algorithms have been proposed, based on the application of an incremental linear PCA method in the feature space \cite{IKPCA2007,Hoegaerts2007}.

To address this limitation, we introduce an online algorithm for real-time applications, called incremental SPEctral Eigenfunction Decomposition (iSPEED). The rank-1 update for the eigendecomposition of the Gram matrix proposed in \cite{hallgren2018incremKPCA} is used as it provides an efficient iterative process compared to similar methods with the added benefit of modularity. Different rank-1 update algorithms can be swapped easily for potential improvements. 

Adaptive methods are often preferred in streaming data environments such as time series analysis for their efficiency and ability to track changes in data statistics. Our approach introduces a novel algorithm that incrementally updates the eigenfunction basis with incoming data while preserving the learning parameters, effectively supporting transfer learning through a change of basis. This iterative process using perturbation theory allows the user to customize their update schedule or frequency, which is significantly more efficient than batch updates. Furthermore, we provide a strategy to reduce the update frequency using a novelty detector by maintaining a compact dictionary using a small subset of the data, resulting in a sparse incremental SPEctral Eigenfunction Decomposition (siSPEED). 

The rest of this paper is organized as follows. In Section \ref{Sec:RKHS}, we briefly introduce the theory of reproducing kernel Hilbet Spaces. In Section \ref{Sec:SPEED}, we derive the spectral embedding using eigenfunctions of the RKHS. Sections \ref{Sec:SPEED-KAF} presents the SPEED KAF algorithms. Section \ref{Sec:Results} shows the experimental results for chaotic time series prediction and compares the novel eigenfunction formulation with existing kernel adaptive filters. Finally, Section \ref{Sec:Conclusion} concludes this paper.

\section{Theory of RKHS}\label{Sec:RKHS}
In 1904, Hilbert introduced his work on kernels and defined what is known as a definite kernel \cite{hilbert1904grundzuge}. Along with Schmidt, who made significant contributions to the development of integral equations and a simplified geometrized theory, they were instrumental in the establishment of functional analysis \cite{Birkhoff1984}. In his thesis, Schmidt proved the existence of eigenvalues and eigenfunctions for the Fredholm integral equation with a continuous symmetric kernel and introduced the Gram-Schmidt orthogonalization \cite{Schmidt1907}. More importantly, Schmidt laid the groundwork for what is now understood as Hilbert space. Hilbert himself never used the term “space,” but Schmidt spoke of “vectors in infinite-dimensional space” and “geometry in a function space" \cite{schmidt1908auflosung, Birkhoff1984}.

Hilbert space later proved to be essential for the formulation of quantum mechanics \cite{prugovecki1982quantum}. Subsequently, Mercer expanded on Hilbert's ideas, leading to his famous Mercer’s theorem in 1909 \cite{mercer1909functions}. Around the same period (1920 to 1922), Banach, Hahn, and Helly introduced concepts that formed another mathematical structure, which became known as Banach space, named by Fréchet \cite{narici2010topological}, which notably, Hilbert space is a subset.

The first work on reproducing kernel Hilbert spaces was by Aronszajn in 1950 \cite{Aronszajn1950}. Later, these concepts were further developed by Aizerman et al. in 1964 \cite{aizerman1964theoretical}. RKHS remained a purely mathematical topic until its application in machine learning with the introduction of kernel support vector machines (SVMs) by Boser et al. (1992) and Vapnik (1995) \cite{boser1992training, vapnik1995nature}. Eigenfunctions, developed for eigenvalue problems involving operators and functions \cite{WilliamsIDD2000}, were also adopted in machine learning \cite{bengio2003spectral} and physics \cite{kusse2006mathematical}. This is closely related to the RKHS, which is a Hilbert space of functions defined by a reproducing kernel, using a weighted inner product \cite{WilliamsIDD2000}.

\begin{thm}[Mercer's Theorem \cite{mercer1909functions}]\label{theorem_Mercer}
	Let kernel $k: \mathcal{X} \times \mathcal{X} \rightarrow \mathbb{R}$ be a symmetric, continuous, positive semi-definite function.
	
	The Hilbert-Schmidt integral operator $T_k$ associated with $k$ is a linear operator on functions $f(\b{x})$ and outputs a new  positive semi-definite function
	\begin{align}\label{equation_Mercer_theorem_T_operator}
		T_k f(\b{x}) \stackrel{\Delta}{=} \int_\mathcal{X} k(\b{x}, \b{y}) f(\b{y})\, d\b{y}
	\end{align}
	which is a Fredholm integral equation \cite{Schmidt1907}, i.e.,
	\begin{align}
		\int\!\!\! \int k(\b{x}, \b{y}) f(\b{y})\, d\b{x}\, d\b{y} \geq 0.
	\end{align}
	Then, there exists a set of orthogonal bases $\{\psi_i(\cdot)\}_{i=1}^\infty$ in $L^2(\mathcal{X})$ consisting of the eigenfunctions of the operator $T_k$, such that the corresponding sequence of eigenvalues $\{\lambda_i\}_{i=1}^\infty$ is nonnegative 
	\begin{align}\label{equation_Mercer_theorem_eigenfunction_decomposition}
		\int k(\b{x}, \b{y})\, \psi_i(\b{y})\, d\b{y} = \lambda_i\, \psi_i(\b{x}).
	\end{align}
	The eigenfunctions corresponding to the non-zero eigenvalues are continuous on $\mathcal{X}$, and $k$ is represented as \cite{aizerman1964theoretical}
	\begin{align}\label{equation_Mercer_theorem_kernel_representation}
		k(\b{x}, \b{y}) = \sum_{i=1}^\infty \lambda_i\, \psi_i(\b{x})\, \psi_i(\b{y})
	\end{align}
	where the convergence is absolute and uniform.
\end{thm}
The eigenfunctions $\psi_i$ form a complete orthogonal system in $L^2(a,b)$, with each corresponding eigenvalue $\lambda_i\geq 0$, which completely characterizes the action of the operator $T_k$.

\subsection{Feature Map or Lifting Function}\label{section_feature_map}

Let $\mathcal{X} \stackrel{\Delta}{=} \{\b{x}_i\}_{i=1}^n$ be a set of data in the original input space, where $\b{x}\in\mathbb{R}^d$. The potentially infinite dimensional feature space or Hilbert space is denoted by $\mathcal{H}$. 

\begin{defn}[Feature Map]
	A feature map is defined as
	\begin{align}
		\b{\phi}: \mathcal{X} \rightarrow \mathcal{H} 
	\end{align}
	which transforms data from the input space to the feature space (Hilbert space), i.e., $\b{x} \mapsto \b{\phi(\x)}$.

	The lifting function $\b{\phi}$ or feature map is potentially infinite-dimensional whose vector elements are \cite{minh2006mercer}:
	\begin{equation}
		\begin{aligned}
			\b{\phi} &= [\phi_1(\b{x}), \phi_2(\b{x}), \dots,  \phi_{n_{\phi}}(\b{x})]^\transpose \\
			&= [\sqrt{\lambda}_1\, \b{\psi}_1(\b{x}), \sqrt{\lambda}_2\, \b{\psi}_2(\b{x}), \dots, \sqrt{\lambda}_{n_{\phi}}\, \b{\psi}_{n_{\phi}}(\b{x})]^\transpose \label{eq:feature_map}
		\end{aligned}
	\end{equation}
	where $\{\b{\psi}_i\}_{i=1}^{n_{\phi}}$ and $\{\lambda_i\}_{i=1}^{n_{\phi}}$ are the eigenfunctions and eigenvalues of the kernel operator $T_k$ in \eqref{equation_Mercer_theorem_eigenfunction_decomposition}, respectively, and  $\b{\phi}\in\mathbb{R}^{n_{\phi}}$, which can be infinite dimensional, i.e., ${n_{\phi}}\leq\infty$.
\end{defn}
	
Combining \eqref{equation_Mercer_theorem_kernel_representation} and \eqref{eq:feature_map}, yields
\begin{align}\label{equation_kernel_inner_product}
	k(\b{x}, \b{y}) = \big\langle \b{\phi}(\b{x}), \b{\phi}(\b{y}) \big\rangle_k = \b{\phi}(\b{x})^\top \b{\phi}(\b{y}) 
\end{align}
i.e., the kernel evaluation between two points in the input space is equivalent to the inner product of the transformed points in the feature space. 

Define the feature matrix as the image of all points in $\mathcal{X}$:
\begin{align}\label{equation_Phi_X_pulled_matrix}
	\b{\Phi} \stackrel{\Delta}{=} [\b{\phi}(\b{x}_1), \b{\phi}(\b{x}_2), \dots, \b{\phi}(\b{x}_n)]
\end{align}
which is $n_{\phi} \times n$ dimensional. 
The kernel or Gram matrix can be calculated as
\begin{align}\label{equation_kernel_inner_product_matrix}
	\mathbf{K} = \big\langle \b{\Phi}, \b{\Phi} \big\rangle_k =  \b{\Phi}^\top \b{\Phi} \in \mathbb{R}^{n\times n}.
\end{align}

\section{Spectral Embedding using Eigenfunctions}\label{Sec:SPEED}
From the theory of RKHS, any solution must lie in the span of all training samples in $\mathcal{H}$. Given $n$ samples, the RKHS is defined as
\begin{align}
	\mathcal{H} = {\rm span}\{\b{\phi}(\b{x}_1), \b{\phi}(\b{x}_2), \dots, \b{\phi}(\b{x}_n)\}.\label{eq:span_data}
\end{align}
By the representer theorem, every function in the RKHS is some combination or weighted sum of these basis functions
\begin{align}
	f &=\sum^n_{i=1}\alpha_i k(\b{x}_i,\cdot)\\
	&= \b{\Phi} \b{\alpha} \label{eq:representertheorem}
\end{align}
i.e., the function $f$ is projected onto a subspace spanned by $\{k(\b{x}_i,\cdot)\}^{n}_{i=1}$. For a detailed discussion and proof of the representation in RKHS please refer to \cite{ghojogh2021RKHS}.

Alternatively, we can use the eigenfunctions of the integral operator to define a lower-dimensional linear approximation
\begin{align}
	\mathcal{H}' = {\rm span}\{\b{\psi}_1(\b{x}), \b{\psi}_2(\b{x}), \dots, \b{\psi}_m(\b{x})\}
\end{align}
where the $m$ most dominant eigenfunctions (corresponding to the $m$ largest eigenvalues) are used as the basis functions to reduce the dimensionality of the feature space, i.e., $m\ll n$ in \eqref{eq:span_data}. Any function in the RKHS can then be decomposed using orthogonal projections onto the subspace spanned by the $m$ eigenfunctions. This construction is the infinite-dimensional analog of the principal component analysis (PCA).

Performing eigendecomposition in $\mathcal{H}$ is intractable, since the eigenfunctions of the integral operator in matrix form is $\frac{1}{n}\b{\Phi}\b{\Phi}^\transpose \in \mathbb{R}^{n_\phi \times n_\phi}$, which is potentially $\infty \times \infty$. However, we can leverage the relationship between the eigendecomposition of the kernel matrix and the integral operator to compute the spectral embedding.

Let $\mathcal{H}$ be the RKHS associated with the kernel function $k$, and $T_\mathcal{H}:\mathcal{H}\rightarrow\mathcal{H}$, an integral operator \cite{williams2000effect, bengio2003out} defined as 
\begin{align}\label{equation_K_operator_integral}
	T_\mathcal{H}f(\b{x}) \stackrel{\Delta}{=} \int_\mathcal{X} k(\b{x},\b{y})\, f(\b{y})\, p(\b{y})\, d\b{y}
\end{align}
where $f \in \mathcal{H}$ and the density function $p$ can be approximated empirically using the sample estimate \cite{williams2000effect}, i.e.,
\begin{align}\label{equation_discrete_kernel_operator}
	T_{n} f(\b{x}) \stackrel{\Delta}{=} \frac{1}{n} \sum_{i=1}^n k(\b{x},\b{x}_i)\, f(\b{x}_i)
\end{align}
which converges to \eqref{equation_K_operator_integral} when $n \rightarrow \infty$, with the operator $T_{n}:\mathcal{H}\rightarrow\mathcal{H}$ and $\b{x}_i$ sampled i.i.d. according to $p$. In the limit, the eigenvectors converge to the eigenfunctions for the linear operator, differing only by a normalization factor. We see that \eqref{equation_Mercer_theorem_T_operator} is a special case of \eqref{equation_K_operator_integral}, when $p$  is uniform in $\mathcal{X}$ and zero outside.
 
\subsection{Spectrum of the Kernel Matrix and the Integral Operator}
Let $(\uplambda_i,\b{\psi}_i)$ denote the $i$th (eigenvalue, eigenfunction) pair of the integral operator $T_\mathcal{H}$, and $(\lambda_i, \b{v}_i)$, the $i$th eigenpair of the Gram matrix $\mathbf{K} \in \mathbb{R}^{n \times n}$, respectively. The two eigensystems are related as follows \cite{bengio2003out, bengio2003spectral, bengio2004learning, rosasco10a}
\begin{align}
	\uplambda_i &= \frac{1}{n}\lambda_i\\
\intertext{and}
	\b{\psi}_i &= \frac{1}{\sqrt[]{\lambda_i}}\b{\Phi}\b{v}_i = \b{\Phi}\underbrace{\frac{1}{\sqrt[]{\lambda_i}}\b{v}_i}_{\b{\alpha}_i}\label{eq:eigenfunction_relation}
\end{align}
where the $i$th eigenfunction lie in the subspace spanned by the data, and the information of the eigenfunction $\b{\psi}_i$ is in the coefficients $\b{\alpha}_i$ of the basis functions, given by the $i$th eigenpair of the Gram matrix, i.e., $\b{\alpha}_i \stackrel{\Delta}{=}\frac{1}{\sqrt[]{\lambda_i}}\b{v}_i\in\mathbb{R}^n$. This elegant solution has the same form as \eqref{eq:representertheorem}.

Now we can approximate \eqref{eq:feature_map} using the spectral embedding feature map
\begin{align}
	\widetilde{\b{\phi}}(\b{x}) \stackrel{\Delta}{=}[\b{\psi}_1(\b{x}),\b{\psi}_2(\b{x}),\cdots,\b{\psi}_m(\b{x})]^\transpose\label{eq:spectral_feature}
\end{align}
where $m$ is set by an appropriate eigenvalue cutoff. Note that the spectral embedding approach is \textit{kernel-independent} (agnostic to the kernel choice and can be easily swapped out for a different kernel) and \textit{data-dependent}, making it a much more versatile representation than the \textit{kernel-dependent} and \textit{data-independent} finite dimensional maps using Taylor series polynomials, Gaussian quadrature, or random Fourier features (depends on sub-Gaussian kernels). The Gaussian approximation is agnostic to the data, as it tries to guarantee convergence for all points (support of the exponential function), making it a much less efficient representation for a particular dataset.

To evaluate \eqref{eq:spectral_feature}, we expand the expression using \eqref{eq:eigenfunction_relation} as
\begin{align}
\widetilde{\b{\phi}}(\b{x}) &= \begin{bmatrix}
	\b{\psi}_1^\transpose\b{\phi}(\b{x}) \\
	\b{\psi}_2^\transpose\b{\phi}(\b{x}) \\
	\vdots \\
	\b{\psi}m^\transpose\b{\phi}(\b{x}) \\\end{bmatrix} = \begin{bmatrix}
	\frac{1}{\sqrt[]{\lambda_1}}\b{v}_1^\transpose\b{\Phi}^\transpose \\
	\frac{1}{\sqrt[]{\lambda_2}}\b{v}_2^\transpose\b{\Phi}^\transpose\\
	\vdots \\
	\frac{1}{\sqrt[]{\lambda_m}}\b{v}_m^\transpose\b{\Phi}^\transpose\\\end{bmatrix} \b{\phi}(\b{x}) \nonumber\\
	&\stackrel{\Delta}{=} \underbrace{\b{\Lambda}^{\text{-}\frac{1}{2}}\b{V}^\transpose}_{\b{\Psi}} \b{\Phi}^\transpose\b{\phi}(\b{x}) = \b{\Psi}\underbrace{\begin{bmatrix}
	k(\b{x}_1,\b{x}),\cdots,
	k(\b{x}_n,\b{x})\end{bmatrix}^\transpose}_{\mathbf{k}_{\b{x}}}\nonumber\\
	&= \b{\Psi}\mathbf{k}_{\b{x}}
  \label{Eq:eigenfunction}
\end{align}
where $\b{\Psi}\in\mathbb{R}^{m\times n}$ is the eigenmap composed using the $m$ most dominant eigenfunctions. This eigenmap is fixed for a given eigendecomposition and is independent of the input. For a given sample $\b{x}$, we can obtain the explicit $m$-dimensional representation by applying the fixed eigenmap to the kernel evaluation vector, $\mathbf{k}_{\b{x}}\in\mathbb{R}^{n}$, between the input sample and the dictionary of $n$ samples used to compute the $n \times n$ Gram matrix. Note the $i$th eigenfunction is the $i$th component of the KPCA of $\b{x}$, up to centering \cite{williams2000effect, bengio2003spectral,rosasco10a}.

\section{SPEctral Eigenfunction Decomposition Kernel Adaptive Filtering (SPEED-KAF)}\label{Sec:SPEED-KAF}
Using the explicit mapping defined by \eqref{Eq:eigenfunction}, we can reformulate kernel adaptive algorithms in the space spanned by the eigenfunctions instead of the sample basis functions. Without loss of generality, we focus on the three most popular kernel adaptive filtering algorithms: the kernel LMS (KLMS) \cite{KLMS}, the kernel RLS (KRLS) \cite{Engel04}, and the Extended-KRLS (Ex-KRLS) \cite{EKRLS}. The same principle can be easily applied to any linear filtering technique to achieve nonlinear solutions.

Once explicitly mapped into the Euclidean subspace spanned by a finite number of eigenfunctions, the linear LMS and RLS algorithms can be directly applied, with constant complexity $O(m)$ and $O(m^2)$, respectively. The SPEED-KLMS is summarized in Algorithm \ref{alg:FS-LMS}, the SPEED-KRLS in Algorithm \ref{alg:FS-RLS}, and the SPEED-Ex-KRLS in Algorithm \ref{alg:FS-ExRLS}. 
\begin{algorithm}[h!]
	\textbf{Initialization:}\\
	$\widetilde{\b{\phi}}(\cdot):\X\rightarrow\R^m$ eigenfunction map\\
	$\b{\omega}_0 = \textbf{0}$: feature space weight vector, $\b{\omega}\in\R^m$\\
	$\eta$: learning rate\\
	\textbf{Computation:}\\
	\For{$i = 1, 2, \cdots$}{
		$e_i = y_{i} -\b{\omega}^\transpose_{i-1}\widetilde{\b{\phi}}(\b{x}_i)$\\
		$\b{\omega}_{i} = \b{\omega}_{i-1}+\,\eta e_i\widetilde{\b{\phi}}(\b{x}_i)$
	}
	\normalsize
	\caption{SPEED-KLMS Algorithm}
	\label{alg:FS-LMS}	
\end{algorithm}

\begin{algorithm}[h!]
	\textbf{Initialization:}\\
	$\widetilde{\b{\phi}}(\cdot):\X\rightarrow\R^m$ eigenfunction map\\
	$\varlambda$: forgetting factor\\
	$\delta$: initial value to seed the inverse covariance matrix $\mathbf{P}$\\
	$\mathbf{P}_0=\delta\mathbf{I}$: where $\mathbf{I}$ is the $D\times D$ identity matrix\\
	$\b{\omega}_0 = \textbf{0}$: feature space weight vector, $\b{\omega}\in\R^m$\\
	\textbf{Computation:}\\
	\For{$i = 1, 2, \cdots$}{
		$e_i = y_{i} -\b{\omega}^\transpose_{i-1}\widetilde{\b{\phi}}(\b{x}_i)$\\
		$\mathbf{g}_{i}=\mathbf{P}_{i-1}\widetilde{\b{\phi}}(\b{x}_i)\left\{\varlambda+\widetilde{\b{\phi}}(\b{x}_i)^\transpose\mathbf{P}_{i-1}\widetilde{\b{\phi}}(\b{x}_i)\right\}^{-1}$\\	
		$\mathbf{P}_{i}=\varlambda^{-1}\mathbf{P}_{i-1}-\mathbf{g}_{i}\widetilde{\b{\phi}}(\b{x}_i)^\transpose\varlambda^{-1}\mathbf{P}_{i-1}$\\
		$\mathbf{w}_{i} = \mathbf{w}_{i-1}+\,\mathbf{g}_{i}e_i$
	}
	\normalsize
	\caption{SPEED-KRLS Algorithm}
	\label{alg:FS-RLS}	
\end{algorithm}

\begin{algorithm}[h!]
	\textbf{Initialization:}\\
	$\widetilde{\b{\phi}}(\cdot):\X\rightarrow\R^m$ eigenfunction map\\
	$\varlambda$: forgetting factor\\
	$\mathbf{A}\in\R^{D\times D}$: state transition matrix\\
	$\delta$: initial value to seed the inverse covariance matrix $\mathbf{P}$\\
	$\mathbf{P}_0=\delta\mathbf{I}$: where $\mathbf{I}$ is the $D\times D$ identity matrix\\
	$\b{\omega}_0 = \textbf{0}$: feature space weight vector, $\b{\omega}\in\R^m$\\
	\textbf{Computation:}\\
	\For{$i = 1, 2, \cdots$}{
		$e_i = y_{i} -\b{\omega}^\transpose_{i-1}\widetilde{\b{\phi}}(\b{x}_i)$\\
		$\mathbf{g}_{i}=\mathbf{A}\mathbf{P}_{i-1}\widetilde{\b{\phi}}(\b{x}_i)\left\{\varlambda+\widetilde{\b{\phi}}(\b{x}_i)^\transpose\mathbf{P}_{i-1}\widetilde{\b{\phi}}(\b{x}_i)\right\}^{-1}$\\	
		$\mathbf{P}_{i}=\mathbf{A}\left\{\varlambda^{-1}\mathbf{P}_{i-1}-\mathbf{g}_{i}\widetilde{\b{\phi}}(\b{x}_i)^\transpose\varlambda^{-1}\mathbf{P}_{i-1}\right\}\mathbf{A}^\transpose+\varlambda q\mathbf{I}{}$\\
		$\mathbf{w}_{i} = \mathbf{A}\mathbf{w}_{i-1}+\,\mathbf{g}_{i}e_i$
	}
	\normalsize
	\caption{SPEED-Ex-KRLS Algorithm}
	\label{alg:FS-ExRLS}
\end{algorithm}

In practice, we would first need an initial batch of data to perform the eigendecomposition of a Gram matrix, before we can construct the finite-dimensional explicit Hilbert space spanned by the eigenfunctions and perform filtering. This evaluation is unsupervised, and only depends on the kernel choice and kernel parameter(s). Since the Gaussian kernel is a measure of similarity between points and distance preserving, we can preprocess the data for a sparse representation ($\rm SPEED_{sparse}$) using a compact dictionary. By only allowing samples exhibiting sufficient novelty (based on a distance metric in the input space) to be included, we can reduce the training size significantly, thus the kernel matrix dimension.

\subsection{Online Eigenfunction and Spectral Decomposition}  

Like KPCA, the finite-dimensional feature map defined in \eqref{Eq:eigenfunction} is a batch method and not suitable for tracking non-stationary input data. Here we present an online algorithm for real-time applications. Adaptive methods are often more desirable for increased time efficiency in streaming data settings with the ability to track changes in statistics. We present a novel algorithm to incrementally update the eigenfunction features while preserving the learning parameters (transfer learning via a change of basis), allowing the user to set the desired update schedule or frequency, based on rank-1 updates to the eigendecomposition of the Gram matrix, which is more computationally efficient than batch methods for iterative updates. Incremental update also improves memory efficiency and enables empirical assessment of the eigendecomposition.

Unlike the covariance matrix in linear PCA, the kernel matrix increases in size with each additional data point. This expansion must be considered, and its impact on the eigensystem needs to be evaluated. The incremental kernel matrix $\mathbf{K}_{n+1}$ created with $n+1$ data examples can be expressed as an expansion and symmetric rank-1 updates to the existing kernel matrix $\mathbf{K}_{n}$. Similarly, the eigendecomposition can be updated using the previous eigensystem. Many existing algorithms have been suggested to perform rank-1 modification to the symmetric eigenproblem \cite{golub1973some, bunch1978rank, dongarra1987fully, sorensen1991orthogonality, gu1994stable, brand2006fast}. Without loss of generality, we apply the rank-1 update algorithm for eigenvalues from \cite{golub1973some} and determine the eigenvectors according to \cite{bunch1978rank}, and previously proposed for incremental KPCA \cite{hallgren2018incremKPCA}.

\subsection{Rank-1 Kernel Eigendecomposition Update}

Given the eigendecomposition of $\mathbf{K}_{n} = \mathbf{V}_n \mathbf{\Lambda}_n \mathbf{V}_n^\transpose$, a rank-1 update can be used to obtain the eigensystem of the incrementally expanded kernel matrix $\mathbf{K}_{n+1} = \mathbf{V}_{n+1} \mathbf{\Lambda}_{n+1} \mathbf{V}_{n+1}^\transpose$.

Denote the kernel evaluation between two points $\b{x}_i$ and $\b{x}_j$ as $k_{i,j} \stackrel{\Delta}{=} k(\b{x}_i, \b{x}_j)$ and the kernel evaluation vector as $\mathbf{k}_{n+1} \stackrel{\Delta}{=} [k_{1,n+1}, k_{2,n+1}, \cdots, k_{n,n+1}]^\transpose$, i.e., a column vector with elements $\{k_{i,n+1}\}_{i=1}^n$ and let
\begin{align}
	\b{\kappa}_1  &\stackrel{\Delta}{=} [\mathbf{k}_{n+1}^\transpose, \frac{1}{2} k_{n+1,n+1} \;]^\transpose\\
	\b{\kappa}_2 &\stackrel{\Delta}{=} [\mathbf{k}_{n+1}^\transpose, \frac{1}{4} k_{n+1,n+1} \; ]^\transpose \\
	\rho &\stackrel{\Delta}{=} 4/k_{n+1,n+1}
\end{align}
then we have
\begin{align} \label{eq:k0update}
	\mathbf{K}_{n+1} &= \begin{bmatrix}
			\mathbf{K}_{n} & \mathbf{0}_n \\
			\mathbf{0}_n^\transpose & \frac{1}{4} k_{n+1,n+1}
		\end{bmatrix}
		+ \rho \b{\kappa}_1 \b{\kappa}_1^\transpose - \rho \b{\kappa}_2 \b{\kappa}_2^\transpose \\
		&\stackrel{\Delta}{=}\mathbf{K}^{0}_{n+1} + \rho \b{\kappa}_1 \b{\kappa}_1^\transpose - \rho \b{\kappa}_2 \b{\kappa}_2^\transpose
\end{align}
corresponding to an expansion of $\mathbf{K}_{n}$ to $\mathbf{K}^{0}_{n+1}$, where $\mathbf{0}_n\in\mathbb{R}^n$ is a column vector of zeros, followed by two rank-1 updates. 

Compared to the eigensystem of $\mathbf{K}_{n}$, the zero-vector expansion $\mathbf{K}^{0}_{n+1}$ will have an additional eigenvalue $\lambda_{n+1} = \frac{1}{4} k_{n+1,n+1}$ and corresponding eigenvector $v_{n+1} = [0,0,\cdots,1]^\transpose$. Since its eigenvalues are all non-negative, $\mathbf{K}^{0}_{n+1}$ is a symmetric, positive semi-definite (SPSD) matrix. It remains SPSD after the first update, as it is the sum of two SPSD matrices ($\b{\kappa}_1\b{\kappa}_1^\transpose$ is a Gram matrix, if each element is viewed as an individual vector). After the second update, the resulting matrix remains SPSD, as this property holds for $\mathbf{K}_{n+1}$.  Without loss of generality, if we use the Gaussian kernel (universal), $k_{n+1,n+1}$ simplifies to 1.

Next, we modify the symmetric eigenvalue problem after a rank-1 perturbation. Given the eigendecomposition of a SPSD matrix $\mathbf{K} = \mathbf{V} \mathbf{\Lambda} \mathbf{V}^\transpose$, let
\begin{align}
	{\mathbf{K}}' &\stackrel{\Delta}{=} \mathbf{V} \mathbf{\Lambda} \mathbf{V}^\transpose + \rho \b{\kappa} \b{\kappa}^\transpose = \mathbf{V} (\mathbf{\Lambda} + \rho \widetilde{\b{\kappa}} {\widetilde{\b{\kappa}}}^\transpose)\mathbf{V}^\transpose\\
	&\stackrel{\Delta}{=}\mathbf{V}' \mathbf{\Lambda}' \mathbf{V}'^\transpose
\end{align}
where $\widetilde{\b{\kappa}} = \mathbf{V}^\transpose \b{\kappa}$, and define the orthogonal decomposition of $\mathbf{\Lambda}' = \mathbf{\Lambda} + \rho \widetilde{\b{\kappa}} {\widetilde{\b{\kappa}}}^\transpose \stackrel{\Delta}{=} \widetilde{\mathbf{V}}\widetilde{\mathbf{\Lambda}}\widetilde{\mathbf{V}}^\transpose $ \cite{bunch1978rank}. Then the eigendecomposition of $\mathbf{K}'$ is given by $\mathbf{V}\widetilde{\mathbf{V}} \widetilde{\mathbf{\Lambda}} \widetilde{\mathbf{V}}^\transpose \mathbf{V}^\transpose$ with the eigenvalues unchanged from $\mathbf{\Lambda}'$ and the eigenvectors $\mathbf{V}' = \mathbf{V}\tilde{\mathbf{V}}$, because the product of two orthogonal matrices is also orthogonal, and the eigendecomposition is unique as long as all the eigenvalues are distinct.

The eigenvalues of $\mathbf{\Lambda}'$ can be calculated in $O(n^2)$ time by finding the roots of the secular equation or characteristic polynomial \cite{golub1973some}
\begin{equation} \label{eq:eigvalupd}
	w(\tilde{\lambda}) \stackrel{\Delta}{=} 1  + \rho \sum_{i=1}^n \frac{{\widetilde{\b{\kappa}}_i}^2}{\lambda_i - \tilde{\lambda}}
\end{equation}
where the modified eigenvalues $\tilde{\lambda}$ are bounded:
\begin{equation} \label{eq:eigbounds}
	\begin{aligned}[c]
		& \lambda_i \le \tilde{\lambda}_i \le \lambda_{i+1}  \\
		& \lambda_n \le \tilde{\lambda}_n \le \lambda_n + \rho \widetilde{\b{\kappa}}^\transpose \widetilde{\b{\kappa}} \\
		& \lambda_{i-1} \le \tilde{\lambda}_i \le \lambda_i \\
		& \lambda_1 + \rho {\widetilde{\b{\kappa}}}^\transpose \widetilde{\b{\kappa}} \le \tilde{\lambda}_1 \le \lambda_1
	\end{aligned}
	\hspace{10pt}
	\begin{aligned}[c]
		i &= 1, 2, \cdots, n-1, \; &\rho > 0 \\
		&&\rho > 0 \\
		i &= 2, 3, \cdots, n, \; &\rho < 0 \\
		&&\rho < 0
	\end{aligned}
\end{equation}
which serves as starting estimates for the root-finding algorithm. It is important to sort the eigenpairs after modifying the eigensystem to ensure the validity of the bounds.

Once the updated eigenvalues have been calculated, the eigenvectors of the perturbed matrix ${\mathbf{K}}'$ are given by \cite{bunch1978rank}
\begin{equation} \label{eq:eigvecupd}
	\mathbf{V}'_i = \frac{\mathbf{V} \mathbf{D}_i^{-1}\widetilde{\b{\kappa}}}{\| \mathbf{D}_i^{-1}\widetilde{\b{\kappa}} \|}
\end{equation}
where $\mathbf{D}_i \stackrel{\Delta}{=} \mathbf{\Lambda} - \tilde{\lambda}_i \mathbf{I}$. The complexity of computing each eigenvector $\mathbf{V}'_i$ is $O(n^2)$ or $O(n^3)$ to update the entire $\mathbf{V}'$.

\subsection{Transfer of Learning Parameters via Change of Basis}
For incremental learning, once the eigenfunctions are updated, the coordinates of the previously learned weights must also be modified, due to the change of basis.

The weight vector $\b{\omega}$ in the subspace spanned by the eigenfunctions can be viewed as the image of a point or function $\b{\phi}(\b{x}_\omega)$ in the RKHS under the linear transformation $\b{\Psi} \b{\Phi}^\transpose$. Instead of trying to find the preimage $\b{\phi}(\b{x}_\omega)$, which is infinite dimensional, we can set $\mathbf{k}_{\b{\omega}}\stackrel{\Delta}{=}\b{\Phi}^\transpose\b{\phi}(\b{x}_\omega)\in\mathbb{R}^{n}$ as the preimage, i.e., $\b{\omega} = \b{\Psi}\mathbf{k}_{\b{\omega}}$.

Since the eigenfunction matrix $\b{\Psi}$ in \eqref{Eq:eigenfunction} has full row rank (the rows are linearly independent, being scaled, transposed eigenvectors of the Gram matrix), the right inverse $\b{\Psi}^{-1}_{\rm right}\stackrel{\Delta}{=} \b{\Psi}^\transpose (\b{\Psi}\b{\Psi}^\transpose)^{-1}$ exists. It has the zero vector only in its left nullspace, and $\b{\Psi}^{-1}_{\rm right}\b{\Psi}$ projects $\mathbb{R}^n$ onto the row space of $\b{\Psi}$. From the learned weight vector $\b{\omega}\in \mathbb{R}^m$ (corresponding to $\b{\Psi}$), we can find the coordinates of the new weight vector $\b{\omega}'\in \mathbb{R}^m$ (under the updated eigenmap $\b{\Psi}'$) using the right inverse by first finding its preimage $\mathbf{k}_{\b{\omega}}$.
\begin{prop}[]
For sets $X$ and $Y$, if a function $f: X\rightarrow Y$ has a right inverse $g: Y\rightarrow X$, then every element of its codomain has a preimage in its domain: $f\circ g = 1_Y \iff f$ is onto $Y$ (subjective).
\end{prop}

\begin{proof}
$\Leftarrow$) Assume $f$ is surjective. Then, $\forall y\in Y$,  $\exists x\in X$ such that $f(x)=y$. Define a function $g \stackrel{\Delta}{=} f^{-1}_{\rm right}: Y\rightarrow X$ (if there is more than one $x$, the inverse function $g$ maps $y$ to an arbitrarily chosen $x$ such that $g$ is well-defined) and it follows:
\begin{align}
	f\circ g(y) = f(g(y)) = f(x_y) = y
\end{align}
and $f\circ g = 1_Y$.

$\Rightarrow$) Assume $f: X\rightarrow Y$ and $g: Y\rightarrow X$ such that $f\circ g =1_Y$, then for each $y\in Y$, the preimage is $x_y\stackrel{\Delta}{=} g(y) \in X$, as  $f(x_y)=f\circ g(y) = 1_Y(y)=y$. Hence, $f$ is surjective.
\end{proof}

Using \eqref{Eq:eigenfunction}, the right inverse simplifies to
\begin{align}
\b{\Psi}^{-1}_{\rm right}&\stackrel{\Delta}{=} \b{\Psi}^\transpose (\b{\Psi}\b{\Psi}^\transpose)^{-1}\\
&= \left(\b{\Lambda}^{\text{-}\frac{1}{2}}\b{V}^\transpose\right)^\transpose \left(\b{\Lambda}^{\text{-}\frac{1}{2}}\b{V}^\transpose\left(\b{\Lambda}^{\text{-}\frac{1}{2}}\b{V}^\transpose\right)^\transpose\right)^{-1}\\
&= \left(\b{V}\b{\Lambda}^{\text{-}\frac{1}{2}}\right) \left(\b{\Lambda}^{\text{-}\frac{1}{2}}\b{V}^\transpose\b{V}\b{\Lambda}^{\text{-}\frac{1}{2}}\right)^{-1}\\
&= \b{V}\b{\Lambda}^{\frac{1}{2}}\label{eq:rightInv}
\end{align}
where the transpose of a diagonal matrix is itself $(\b{\Lambda}^{\text{-}\frac{1}{2}})^\transpose = \b{\Lambda}^{\text{-}\frac{1}{2}}$ and the inverse of an orthogonal matrix is its transpose $\b{V}^\transpose\b{V} = \mathbf{I}$. We see that the right inverse can be easily obtained without having to perform any inversion due to the properties of eigendecomposition. From \eqref{eq:rightInv}, the preimage $\mathbf{k}_{\b{\omega}}$ can be computed as
\begin{align}
	\b{\Psi}\mathbf{k}_{\b{\omega}} &= \b{\omega}\\
	\b{\Psi}^{-1}_{\rm right}\b{\Psi}\mathbf{k}_{\b{\omega}} &=\b{\Psi}^{-1}_{\rm right}\b{\omega}\\
	\b{V}\b{\Lambda}^{\frac{1}{2}}\b{\Lambda}^{\text{-}\frac{1}{2}}\b{V}^\transpose\mathbf{k}_{\b{\omega}}&=\b{V}\b{\Lambda}^{\frac{1}{2}}\b{\omega}\\
	\mathbf{k}_{\b{\omega}}&=\b{V}\b{\Lambda}^{\frac{1}{2}}\b{\omega}.\label{eq:preimage}
\end{align}

The new coordinates of the weight vector $\b{\omega}'$, corresponding to the explicit space spanned using the updated eigenfunctions with $n+1$ samples, is equal to the updated eigenfunction matrix $\b{\Psi}'\in\mathbb{R}^{m\times (n+1)}$ applied to the preimage $\mathbf{k}'_{\b{\omega}}\stackrel{\Delta}{=}\b{\Phi}'\phi(\b{x}_\omega)$, i.e., $\b{\omega}' = \b{\Psi}'\mathbf{k}'_{\b{\omega}}$. Since the preimage $\mathbf{k}_{\b{\omega}}$ obtained from $\b{\omega}$ in \eqref{eq:preimage} has only $n$ components, i.e., missing the term $k(\b{x}_{n+1},\b{x}_\omega)$, we can either pad it with a zero to account for the missing $(n+1)$th sample basis function or equivalently, truncating $\b{\Psi}'$ to remove its last column, i.e., $\widehat{\b{\omega}}' = \b{\Psi}'(n)\b{\Psi}^{-1}_{\rm right}\b{\omega}$, where $\b{\Psi}'(n)\in\mathbb{R}^{m\times n}$ denotes the submatrix by deleting the $(n+1)$th column of $\b{\Psi}'$. This provides a stable solution, albeit the converge rate is slower since contribution of the $(n+1)$th sample has to be learned from scratch.

Alternatively, we can approximate the missing kernel function with its nearest neighbor evaluation, i.e. $	\mathbf{k}_{\b{\omega}}(i^*) = k(\b{x}_{i^*},\b{x}_\omega) \approx k(\b{x}_{n+1},\b{x}_\omega)$, where $i^* = \displaystyle\argmin_{1\leq i\leq n} \norm{\b{x}_i-\b{x}_{n+1}}^2$.

The algorithm for incremental eigendecomposition and transformation of the weight vector is summarized in Algorithm \ref{al:iSPEED}, using an auxiliary function $\mathrm{rank1update}$ that takes the previous eigensystem and updates it for a rank-1 additive perturbation.
	
\begin{algorithm}
	\label{al:iSPEED}
	\caption{Incremental SPEctral Eigenfunction Decomposition (iSPEED)}
		\textbf{Initialization:}\\
		$\{\b{x}_i\}_{i=1}^{n}$: dataset\\
		$m$: explicit space dimension\\
		$k(\cdot,\cdot)$: kernel function\\
		$\mathbf{\Lambda}$: eigenvalues of $\mathbf{K}_{n}$\\
		$\mathbf{V}$: eigenvectors of $\mathbf{K}_{n}$\\
		$\b{\Psi}$: eigenmap of $\mathbf{K}_{n}$\\
		$\b{\omega}$: weight vector in $\mathbb{R}^m$\\
		\textbf{Computation:}\\
		\For{$\b{x}_{n+1}$}{
		$\mathbf{\Lambda} \leftarrow [\mathbf{\Lambda}, k(\b{x}_{n+1},\b{x}_{n+1}) / 4]$\\
		$\mathbf{V} \leftarrow \begin{bmatrix} \mathbf{V} & 0 \\ 0 & 1 \end{bmatrix} $\\
		 $\rho \leftarrow 4 / k_{n+1,n+1}$\\
		 $\b{\kappa}_1 \leftarrow [k_{1,n+1}, k_{2,n+1}, \cdots k_{n+1,n+1}/2]$\\
		 $\b{\kappa}_2 \leftarrow [k_{1,n+1}, k_{2,n+1}, \cdots k_{n+1,n+1}/4]$\\
		 $\mathbf{\Lambda}, \mathbf{V} \leftarrow \mathrm{rank1update(}\mathbf{\Lambda}, \mathbf{V}, \rho, \b{\kappa}_1 \mathrm{)}$\\
		 $\mathbf{\Lambda}, \mathbf{V} \leftarrow \mathrm{rank1update(}\mathbf{\Lambda}, \mathbf{V}, -\rho, \b{\kappa}_2 \mathrm{)}$\\
		 Sort the eigenvectors by their eigenvalues.\\
		 $\b{\Psi} \leftarrow [\frac{1}{\sqrt[]{\lambda_1}}v_1^\transpose, \frac{1}{\sqrt[]{\lambda_2}}v_2^\transpose, \cdots \frac{1}{\sqrt[]{\lambda_m}}v_m^\transpose]^\transpose$\\
		 $\b{\omega} \leftarrow \b{\Psi}'(n)\b{\Psi}^{-1}_{\rm right}\b{\omega}$\\
		 \ \qquad or $\b{\Psi}'[\b{\Psi}^{-1}_{\rm right}\b{\omega};\mathbf{k}_{\b{\omega}}(i^*)]$\\
		 \ \qquad where  $i^* = \displaystyle\argmin_{1\leq i\leq n} \norm{\b{x}_{n+1}-\b{x}_i}^2$}
\end{algorithm}

\subsection{Sparse Update}
The complexity of the online update is less than that of the batch update, but is still significant. To further reduce the computational cost, we can combine the online update with a compact dictionary and reduce the update frequency. Only data samples that are sufficiently far away from existing samples in the dictionary are admitted and triggers an update. For stationary data, the eigendecomposition is expected to stabilize over time: as more data is observed, novelty decreases, thus diminishing the need for updates. More elaborate sparsification algorithms can also be used, such as the Nearest Instance Centroid-Estimation (NICE) approach \cite{NICE}. Algorithm \ref{alg:CD-SPEED} outlines the distance-based novelty detection update procedure.

\begin{algorithm}[h!]
	\textbf{Initialization:}\\
	$\b{\mathcal{D}}=\{\b{x}_i\}_{i=1}^{N}$: initial dictionary of training data\\
	$d_{\rm th}$: distance threshold\\
	$\b{\Psi}$: eigenmap of the kernel matrix on $\b{\mathcal{D}}$\\
	$\b{\omega}$: feature space weight vector corresponding to $\b{\Psi}$\\
	\textbf{Computation:}\\
	\For{$i = 1, 2, \cdots$}{
		$\displaystyle d_{\min}=\min_{1\leq j\leq |\b{\mathcal{D}}|} \norm{\b{x}_i-\b{x}_j}^2$\\
		\If{$d_{\min}\geq d_{\rm th}$}{
			$\b{\mathcal{D}}=\{\b{\mathcal{D}},\b{x}_i\}$\\
			Update $\b{\Psi}$ and $\b{\omega}$ using iSPEED (Algorithm \ref{al:iSPEED}).
		}
	}
	\normalsize
	\caption{Sparse Incremental SPEctral Eigenfunction Decomposition (siSPEED)}
	\label{alg:CD-SPEED}
\end{algorithm}

\section{Simulation Results}\label{Sec:Results}
To demonstrate the effectiveness of the eigenfunction-based algorithms proposed here, we perform one-step ahead prediction on the Mackey-Glass (MG) chaotic time series \cite{Mackey77}, defined by the delay differential equation
\begin{align*}
	\frac{d y_t}{d t} =\frac{\beta y_{(t-\tau)}}{1+y^{n}_{(t-\tau)}} -\gamma y_t
\end{align*}
where $\beta=0.2$, $\gamma=0.1$, $\tau=30$, $n=10$, discretized at a sampling period of 6 seconds using the fourth-order Runge-Kutta method, with initial condition $y_0 = 0.9$. Chaotic dynamics are highly sensitive to initial conditions, where even a small change in the current state can lead to vastly different outcomes over time. This makes long-term prediction intractable and is commonly known as the butterfly effect \cite{Ott02}. 

Additive white Gaussian noise with standard deviation of 0.02 is introduced to the MG time series. The data are standardized by subtracting the mean then dividing by its standard deviation, followed by diving the resulting maximum absolute value to guarantee the sample values fall within the range of $[-1,1]$ (for managing numerical error in approximation methods like Taylor series expansion). A time embedding (input dimension) of $d=7$ is used. Results are averaged over 100 independent trials (unless specified otherwise), with each training set consisting of 2000 consecutive samples from a random starting point in the time series, and the test set consists of 200 consecutive samples located in the future (at least 200 time steps away from the last training sample). 

\subsection{Kernel Principal Component Convergence and Accuracy}
To evaluate the KPCA performance, we assume that the data is centered. First, we demonstrate the fidelity of the Gram matrix reconstruction using eigenmaps in \eqref{eq:spectral_feature}, i.e., $\hat{\mathbf{K}} \stackrel{\Delta}{=}\widetilde{\b{\Phi}}(\b{X})^\top \widetilde{\b{\Phi}}(\b{X})\approx\mathbf{K}$. Fig. \ref{fig:Fnorm} shows the mean normalized Frobenius norm \begin{align}
	\norm{\mathbf{K} - \hat{\mathbf{K}}}_{\rm F} \stackrel{\Delta}{=} \sqrt{\frac{1}{mn}\sum_{i}^m\sum_{j}^n (|k_{ij}-\hat{k}_{ij}|/|k_{ij}|)^2}
\end{align} of the difference between a $500\times 500$ Gram matrix using the Gaussian kernel evaluation (with kernel parameter $\sigma = 1$) versus the reconstruction using the dot products in the explicit feature space as a function of the number of eigenfunctions used, averaged over 100 runs. Since sparsification yields a smaller dictionary size than the total available number of samples (500) in this experiment, for ${\rm SPEED}_{\rm sparse}$, we averaged the normalized Frobenius norm up to the minimum dictionary size (206) across trials. From Fig. \ref{fig:Fnorm}, we see that the eigenfunction approximations provide a high degree of accuracy even for small number of basis functions.

\begin{figure}[t!]
	\centering
	\includegraphics[width=0.4\textwidth]{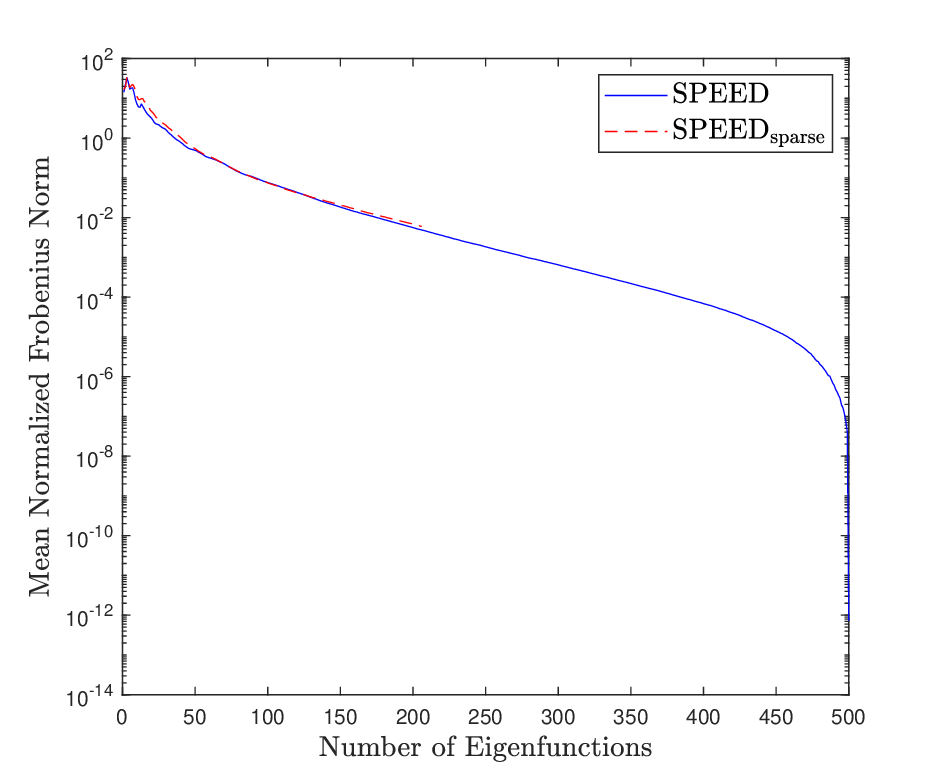}
	\caption{Normalized Frobenius norm of the difference between the Gram matrix $\mathbf{K}_{500}$ and the reconstructions using dot products of the eigenfunction features and the sparse eigenfunction features (distance threshold of 0.06), as a function of the number of eigenfunctions used, averaged across 100 runs.}
	\label{fig:Fnorm}
\end{figure}

Next, to compare the performance of the sparse incremental method (siSPEED) with the batch sparse method ${\rm SPEED}_{\rm sparse}$, we define the following distance measure between the principle axes angles \cite{IKPCA2007}
\begin{align}
	d_\theta(\mathcal{F},\mathcal{F}')\stackrel{\Delta}{=}\sqrt{\sum^m_{i=1}\theta^2_i}
\end{align}
where $\mathcal{F}$ and $\mathcal{F}'$ are $m$-dimensional subspaces in the RKHS and ${\theta_1,\cdots,\theta_r}$ are the principal angles between them, with $0\leq d_\theta \leq m\pi/2$. The subspace dimension $m$ was set to 20. The sparse incremental method was initialized using the first 100 candidate samples out of 2000 and the novelty or distance threshold was set to 0.06. Fig. \ref{fig:dist_angle} shows the subspace angular distance as a function of the number of sparse updates. We see that the initial subspace (batch of 100) converge toward the ground truth subspace (batch of 2000) without any drift.

\begin{figure}[t!]
	\centering
	\includegraphics[width=0.4\textwidth]{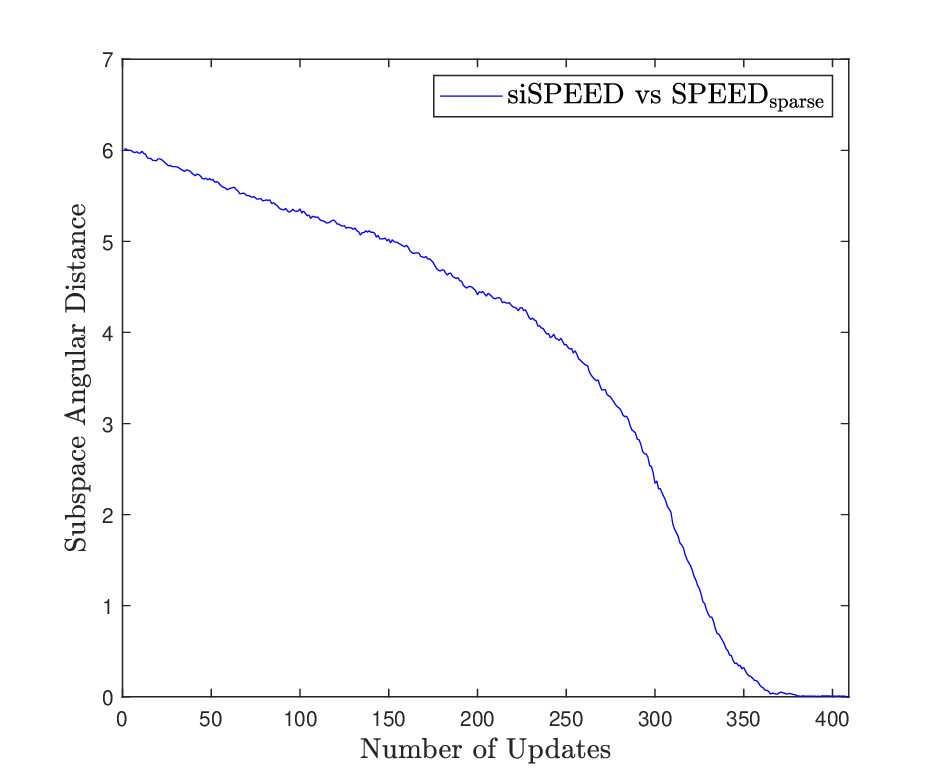}
	\caption{Mean subspace angular distance between the sparse incremental method and the batch method averaged over 500 trials.}
	\label{fig:dist_angle}
\end{figure}

\subsection{Explicit Feature Comparison}
In the first kernel adaptive filtering experiment, we compare the performances of well-known explicit feature maps using the mean squared error (MSE). The feature space or finite-dimensional RKHS dimension is set to $D = 330$, corresponding to Taylor polynomials of degrees up to 4, for input dimension $d = 7$. Two types of RFFs are randomly generated in each trial: RFF1 (sine-cosine pair) and RFF2 (cosine with nonshift-invariant phase noise). An $8$th-degree quadrature rule with sub-sampling is chosen to generate the explicit GQ feature mapping and is fixed across all trials. Similarly, the TS expansion mappings are completely deterministic and fixed for all trials. For the quantized KLMS (QKLMS) algorithm, an appropriate quantization factor of $q=0.06$ is selected to derive a comparable number of centers as the dimension of the explicit features.

Learning curves for the different types of KAF algorithms are plotted in Fig. \ref{fig:MG_pred} for a learning rate of $\eta = 0.1$, with the linear LMS as a baseline. The fixed dimensions are indicated in parentheses next to the explicit features used. The KLMS algorithm grows linearly with the number of training samples, so its final size is 2000. The QKLMS uses a subset of the samples that are sufficiently far away from each other, which results in 398 centers with the quantization factor used. For SPEED and the sparse SPEED features using a compact dictionary, the parenthesized dimension is presented as a fraction with the numerator indicating the feature space dimension $m$ or the number of eigenfunctions used as basis, and the denominator indicates the number of sample basis $n$ used to generate the feature map in \eqref{Eq:eigenfunction}.

We observe that a simple deterministic feature, such as Taylor series expansion, can outperform random features, and the linear LMS algorithm using GQ features outperformed all kernel-based finite-dimensional RKHS filters. The performance of data-based SPEED features exceeded that of all other explicit feature maps, while using 1 order of magnitude fewer features (30 vs. 330). This also demonstrates that SPEED has greater generalization capability and can reduce overfitting.

\begin{figure}[t!]
	\centering
	\includegraphics[width=0.4\textwidth]{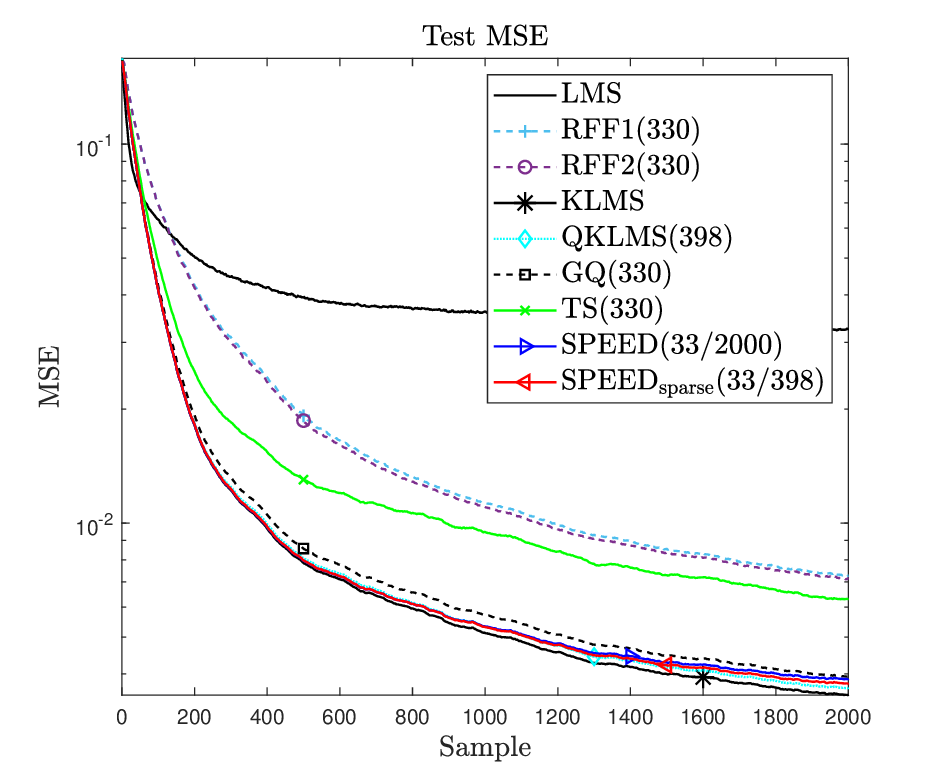}
	\caption{Learning curves averaged over 100 independent runs. Using eigenfunction features, SPEED and the sparse SPEED outperform other explicit features while using 1 order of magnitude fewer dimensions (33 vs. 330).}
	\label{fig:MG_pred}
\end{figure}

Fig. \ref{fig:LC_EFnumber} shows the final test performance as a function of the number of eigenfunctions used for SPEED and sparse SPEED. We observe that using 50 eigenfunctions, a linear LMS outperforms even the QKLMS algorithm and is comparable to the KLMS algorithm using 2000 sample basis functions.

 \begin{figure}[t!]
	\centering
	\includegraphics[width=0.4\textwidth]{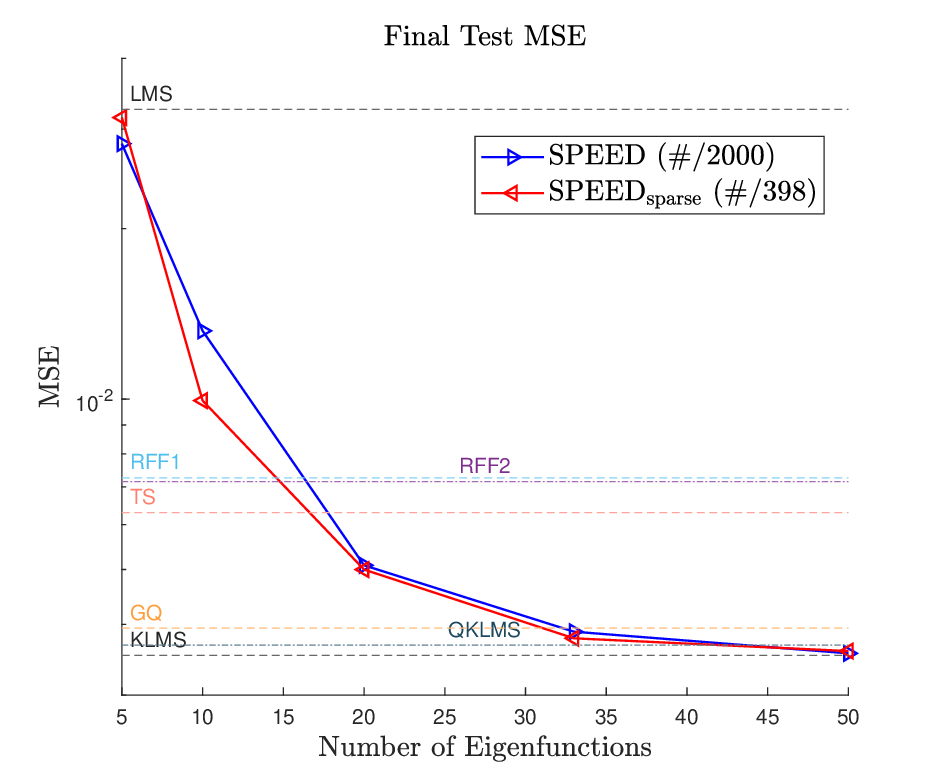}
	\caption{Final test MSE after training as a function of the number of eigenfunctions used for SPEED and sparse SPEED.}
	\label{fig:LC_EFnumber}
\end{figure}

Fig. \ref{fig:LC_batchsize} shows the final test MSE after training as a function of the batch size used for constructing the eigenmaps. Using 50 eigenfunctions, we observe that a small batch of the initial 100 samples is sufficient to outperform GQ features, and a batch of 500 samples outperformed all but the KLMS algorithm.
\begin{figure}[t!]
 	\centering
 	\includegraphics[width=0.4\textwidth]{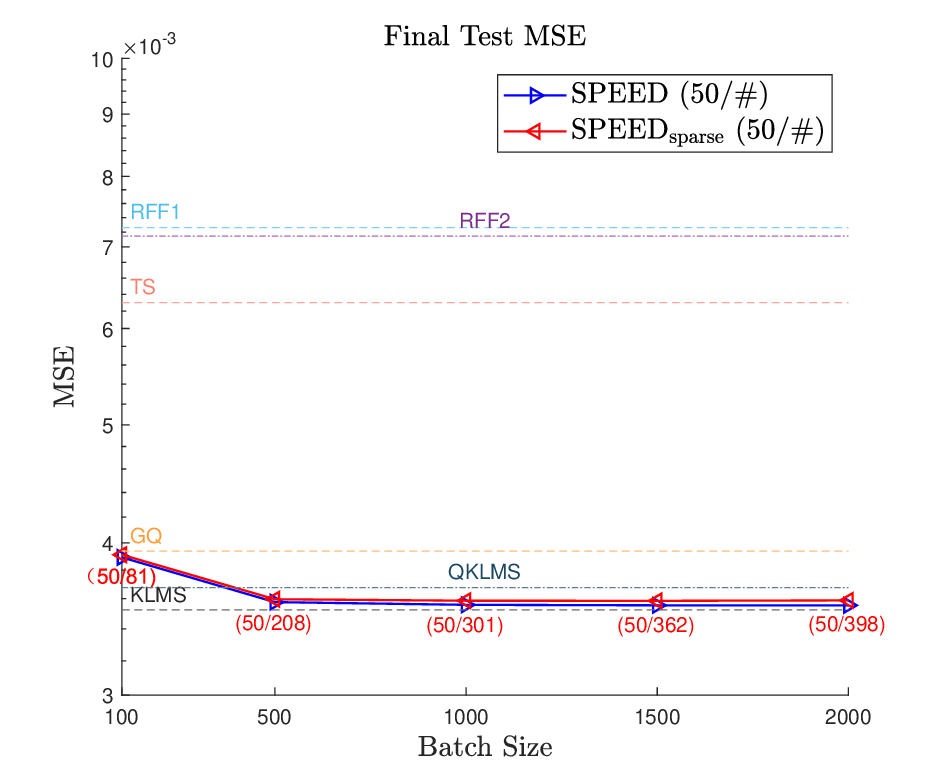}
 	\caption{Final test MSE after training as a function of the batch size used to decompose the eigenfunctions.}
 	\label{fig:LC_batchsize}
 \end{figure}
 
Fig. \ref{fig:KRLS} shows the test performance using the KRLS method. Again, we see the superior performance of the SPEED method. 
 
 \begin{figure}[t!]
 	\centering
 	\includegraphics[width=0.4\textwidth]{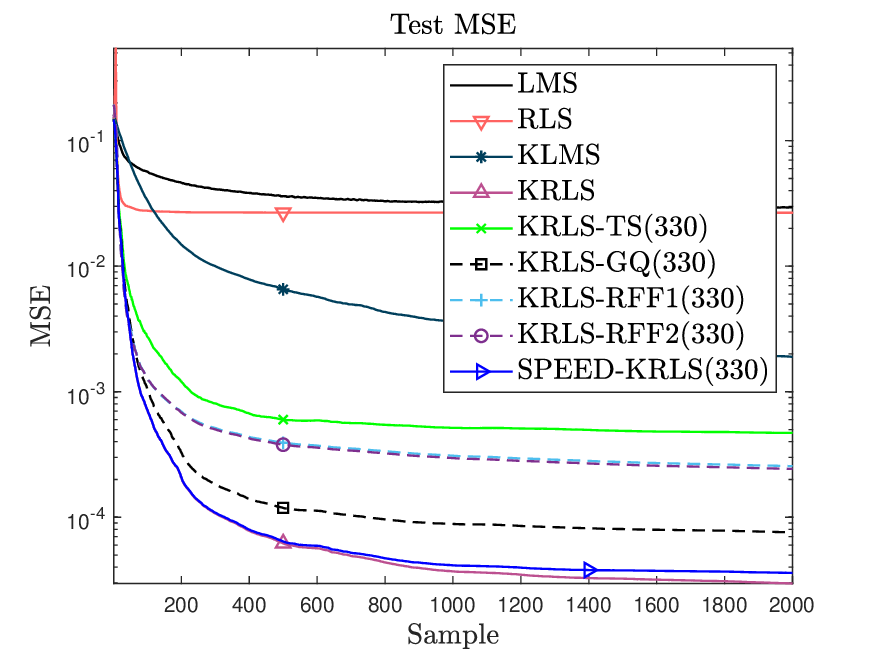}
 	\caption{KRLS test performances averaged over 100 runs.}
 	\label{fig:KRLS}
 \end{figure}
 
 \subsection{Online SPEctral Eigenfunction Decomposition (SPEED)}
 Next we demonstrate the capability of the sparse incremental algorithm. We seed siSPEED with the first 100 samples, then allow it to update based on the novelty criterion. Fig. \ref{fig:iSPEED} shows the continual learning performance of the algorithm as new samples are added to the eigenfunction decomposition. It is able to outperform the sparse SPEED using a fixed batch of 100 samples and eventually converge to the sparse SPEED and SPEED performance using all 2000 samples.
 \begin{figure}[t!]
	\centering
	\includegraphics[width=0.4\textwidth]{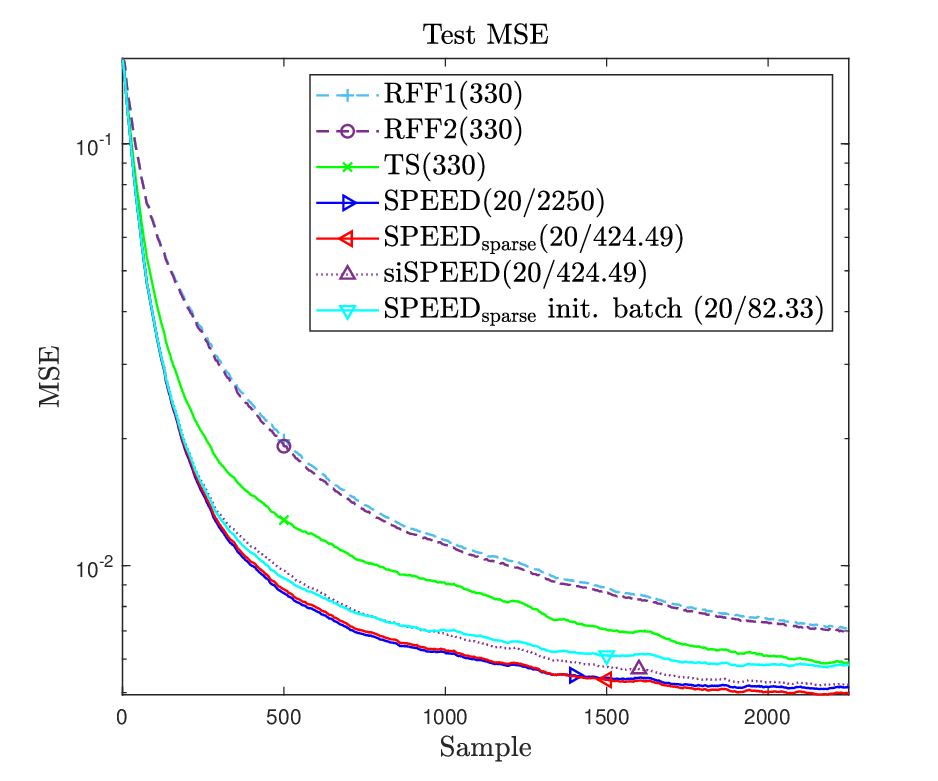}
	\caption{Test performance using sparse incremental SPEED.}
	\label{fig:iSPEED}
\end{figure}

\section{Conclusion}\label{Sec:Conclusion}
We presented a novel approach to simultaneously solve both dimensionality problems plaguing kernel adaptive filtering methods, by providing a Euclidean representation of the RKHS while reducing its dimensionality in a principled way. We also showed an efficient online algorithm for updating the eigensystem in the RKHS, which allows us to track the kernel eigenspace dynamically. Simulation results show that this framework is robust and can outperform similar methods, ideally suited for nonlinear adaptive filtering.

In the future we will further reduce the computational complexity of the explicit Hilbert space construction using novel dimensionality-reduction techniques and manifold learning. We will also apply this framework to other advanced signal processing techniques and model more complex signals and systems. A major advantage of the kernel method is that it operates on functions in the RKHS and changing the kernel function does not impact the underlying learning algorithm. One is free to choose the input representation, with an appropriate reproducing kernel. This opens the door to  solutions using neural recordings or spike trains.

\bibliographystyle{IEEEtran}
\bibliography{IEEEabrv,references}{}
\end{document}